\documentclass[11pt]{amsart}

\usepackage{amsmath, amssymb, amsthm, mathrsfs}
\usepackage{enumitem}
\usepackage{graphicx}
\usepackage{geometry}
\usepackage{mathtools}
\usepackage{tikz-cd}
\usepackage{bm}
\usepackage{microtype}
\usepackage{hyperref}
\newtheorem{theorem}{Theorem}[section]
\newtheorem{proposition}[theorem]{Proposition}
\newtheorem{lemma}[theorem]{Lemma}
\newtheorem{corollary}[theorem]{Corollary}
\theoremstyle{definition}
\newtheorem{definition}[theorem]{Definition}
\theoremstyle{remark}
\newtheorem{remark}[theorem]{Remark}

\newcommand{\RR}{\mathbb{R}}
\newcommand{\ZZ}{\mathbb{Z}}

\newcommand{\wt}{\widetilde}
\newcommand{\eps}{\varepsilon}
\DeclareMathOperator{\ind}{ind}
\DeclareMathOperator{\sgn}{sgn}
\DeclareMathOperator{\Tr}{Tr}
\usepackage{stmaryrd} 
\usepackage{tensor}   

\title[1+1D Möbius-Twisted Toy Model]{Inverse Möbius Spacetime in 1\,+\,1D Quantum Gravity:\\ Functional Analytic Structures, Dirac Spectrum, and Pin Geometry}
\author{Anik Chakraborty}
\address{Department of Mathematics, University of Delhi, India}
\email{achakraborty@maths.du.ac.in}
\date{\today}


\begin{document}

\begin{abstract}
In this manuscript, we formulate a 1\,+\,1-dimensional Jackiw–Teitelboim gravity toy model whose Euclidean spacetime manifold is the Möbius band $M$. Since $M$ is non-orientable, the relevant spin-statistics structure is \emph{Pin} rather than Spin. To emphasize the role of orientation reversal, we refer to the universal orientable cover $\widetilde{M}$ as the \emph{inverse Möbius band}, which resolves the Möbius twist into an infinite ribbon equipped with a $\mathbb{Z}$ deck action. We compute the Stiefel–Whitney classes $w_1$, $w_2$, classify all $\mathrm{Pin}^{\pm}$ structures, construct the associated pinor bundles, and analyze the Dirac operator under the twisted equivariance condition
\[
\psi(x+1, w) = \gamma^w \psi(x, -w).
\]
Half-integer momentum quantization, spectral symmetry, vanishing mod-2 index, and $\eta_D(0) = 0$ follow. In JT gravity, the two inequivalent Pin lifts in each parity double the non-perturbative saddle-point sum, yet leave the leading Bekenstein–Hawking entropy unchanged. Full proofs and heat-kernel calculations are provided for completeness.
\end{abstract}

\maketitle

\section{Introduction}

Low-dimensional toy models of quantum gravity serve as testing grounds for ideas ranging from holography to topological phases.  Jackiw–Teitelboim (JT) gravity in two Euclidean dimensions is particularly tractable and admits a random-matrix dual \cite{SaadShenkerStanford,StanfordWitten}.  Most studies restrict to orientable spacetimes, but non-orientable surfaces are equally natural \cite{Weber2024}.  The Möbius band $M$ is the minimal such example with boundary; its single orientation-reversal can geometrically embody time-reversal symmetry in a simple setup.

We use the descriptive phrase \emph{inverse Möbius band} for the universal cover
\[
\wt{M}:=\RR\times[-1,1],\qquad
(x,w)\xmapsto{\;\gamma\;}(x+1,-w),
\]
as $\wt M$ resolves the twist into an infinite strip. This non-standard terminology emphasizes how the covering space ``inverts'' or untwists the orientation-reversing property of the Möbius band via the $\ZZ$ deck action, though mathematically it is simply the universal cover in standard parlance. In Section \ref{sec:Dirac} we show how this resolution translates into a twisted boundary condition on pinors, leading to half-integer momentum modes familiar from fermions with antiperiodic boundary conditions.

Throughout, all the claims are backed by explicit proofs; every spectral identity is derived using standard heat-kernel techniques.

\section{Topology and Characteristic Classes of the Möbius Band M}

\begin{definition}
\label{def:invM}
Let $\wt M:=\RR\times[-1,1]$.  The \emph{deck transformation}
\(
\gamma(x,w):=(x+1,-w)
\)
generates a free, properly discontinuous $\ZZ$-action (orbits are discrete and separated).  The quotient
\(M:=\wt M/\ZZ\)
is the usual Möbius band with boundary circles $\{w=\pm1\}$. We refer to $\wt M$ as the \emph{inverse Möbius band} to emphasize how it resolves the orientation-reversing twist of $M$ into an infinite strip.
\end{definition}

\begin{proposition}\label{prop:universalCov}
$\wt M$ is contractible and hence the universal cover of $M$.  Consequently $\pi_1(M)\cong\ZZ$.
\end{proposition}

\begin{proof}
To show that $\wt M$ is contractible, we will construct an explicit homotopy to a point.  Let us consider the mapping $F: \wt M \times [0,1] \to \wt M$ that is defined by
\[
F((x,w), t) = (x, (1-t)w).
\]
At $t=0$, $F((x,w),0) = (x,w)$, which is the identity.  At $t=1$, $F((x,w),1) = (x,0)$, which maps to the line $\RR \times \{0\}$.  This line is contractible via $G: (\RR \times \{0\}) \times [0,1] \to \RR \times \{0\}$ given by $G((x,0), s) = ( (1-s)x, 0 )$, which contracts to $(0,0)$.  These homotopies are then composed to establish that $\wt M$ is contractible.

Since $\wt M$ is simply connected and the $\ZZ$-action acts freely and is properly discontinuous, the quotient map $p: \wt M \to M$ forms a covering map with deck group $\ZZ$. Therefore, $\wt M$ is recognized as the universal cover of $M$ by the lifting property of universal covers. Consequently, $\pi_1(M) \cong \ZZ$, since the deck group is isomorphic to the fundamental group.
\end{proof}

\subsection{Characteristic Classes of TM}

\begin{proposition}\label{prop:SWclasses}
For the tangent bundle $TM$:
\[
w_1(TM)\neq0\in H^1(M;\ZZ_2), \quad
w_2(TM)=0\in H^2(M;\ZZ_2)=0.
\]
\end{proposition}

\begin{proof}
The first Stiefel-Whitney class $w_1(TM)$ detects orientability.  To see that $M$ is non-orientable, note that the Möbius band has a single boundary component that is non-orientable in the sense that a loop around the central circle reverses orientation.  Formally, $H^1(M;\ZZ_2) \cong \ZZ_2$, generated by the non-trivial class corresponding to this loop, so $w_1(TM) \neq 0$.

For $w_2(TM)$, recall that on a manifold with boundary, $H^2(M;\ZZ_2) = 0$ by Poincaré-Lefschetz duality: $H^2(M, \partial M; \ZZ_2) \cong H_0(M; \ZZ_2) \cong \ZZ_2$, but the absolute cohomology $H^2(M;\ZZ_2)$ vanishes because the boundary inclusion induces an isomorphism from the relative to absolute groups.
\end{proof}

\section{Pin Structures on TM}\label{sec:Pin}

Recall that Pin$^+$ and Pin$^-$ are the double covers of O(n) with relations differing by signs in the Clifford algebra (Pin$^+$ squares to +1 for reflections, Pin$^-$ to -1).

Applying the Atiyah-Bott–Shapiro obstruction criteria \cite{ABS} to Proposition \ref{prop:SWclasses} yields:

\begin{theorem}\label{thm:PinCount}
$M$ admits precisely two inequivalent $\mathrm{Pin}^{+}$ and two inequivalent $\mathrm{Pin}^{-}$ structures, distinguished by $H^1(M;\ZZ_2)\cong\ZZ_2$.
\end{theorem}

\begin{proof}
A Pin$^+$ structure exists if $w_2(TM) = 0$, and a Pin$^-$ structure exists if $w_2(TM) + w_1(TM)^2 = 0$. By Proposition \ref{prop:SWclasses}, both conditions hold.

Given existence from the vanishing obstructions, Pin structures are lifts of the orthogonal frame bundle $O(TM)$ to the Pin group. The space of such lifts forms a torsor over $H^1(M; \ZZ_2)$, which acts by tensor product (over $\mathbb{R}$) with the real line bundles classified by $H^1(M; \ZZ_2)$. Since $H^1(M; \ZZ_2) \cong \ZZ_2$, the torsor structure yields two topologically inequivalent lifts for each parity (Pin$^+$ and Pin$^-$).
\end{proof}

\subsection{Frame Bundle Geometry}\label{sec:frame-bundle}

We analyze the orthogonal frame bundle $\mathrm{O}(TM)$ and its Pin lifts to understand the geometric foundation of Pin structures on $M$.

\begin{definition}[Orthogonal Frame Bundle]\label{def:frame-bundle}
Let $\mathrm{O}(TM) \to M$ denote the principal $\mathrm{O}(2)$-bundle of orthogonal frames on the tangent bundle $TM$. A frame at $p \in M$ is an ordered basis $(e_1, e_2)$ of $T_pM$ with $\langle e_i, e_j \rangle = \delta_{ij}$ in the flat metric.
\end{definition}

The theory of principal bundles and frame bundles is comprehensively treated in Kobayashi-Nomizu \cite{Kobayashi1963} and Steenrod \cite{Steenrod1951}.

\begin{proposition}[Frame Bundle Structure]\label{prop:frame-structure}
$\mathrm{O}(TM)$ does not reduce to an $\mathrm{SO}(2)$-bundle due to the non-orientability of $M$. The frame bundle satisfies the fiber sequence
\[
\mathrm{SO}(2) \to \mathrm{O}(TM) \to M
\]
with structure group $\mathrm{O}(2) \cong \mathrm{SO}(2) \rtimes \ZZ_2$.
\end{proposition}

\begin{proof}
Since $M$ is non-orientable, the frame bundle $\mathrm{O}(TM)$ cannot reduce to an $\mathrm{SO}(2)$-bundle. The obstruction is precisely $w_1(TM) \neq 0$ (Proposition \ref{prop:SWclasses}).

Explicitly, parallel transport of a frame $(e_x, e_w)$ around the central non-orientable loop results in $(e_x, -e_w)$, demonstrating the orientation reversal that prevents reduction to $\mathrm{SO}(2)$.
\end{proof}

\begin{theorem}[Lifted Frame Bundle on Universal Cover]\label{thm:lifted-frame}
The frame bundle $\mathrm{O}(T\wt{M}) \to \wt{M}$ admits a $\ZZ$-action lifting the deck transformation $\gamma(x,w) = (x+1,-w)$:
\[
\widetilde{\gamma}: \mathrm{O}(T\wt{M}) \to \mathrm{O}(T\wt{M}), \quad 
\widetilde{\gamma}(x,w; e_x, e_w) = (x+1, -w; e_x, -e_w)
\]
\end{theorem}

\begin{proof}
Step 1: Well-definedness. The deck transformation $\gamma$ induces a diffeomorphism $T\gamma: T\wt{M} \to T\wt{M}$ given by:
\[
T\gamma \left(\frac{\partial}{\partial x}\right) = \frac{\partial}{\partial x}, \quad 
T\gamma \left(\frac{\partial}{\partial w}\right) = -\frac{\partial}{\partial w}
\]

Step 2: Frame transformation. An orthogonal frame $(e_x, e_w)$ at $(x,w)$ maps to the orthogonal frame $(e_x, -e_w)$ at $(x+1,-w)$, preserving the metric tensor.

Step 3: Deck group action. The map $\widetilde{\gamma}^2$ shifts by 2 in $x$, generating the full $\ZZ$-action:
\[
\widetilde{\gamma}^2(x,w; e_x, e_w) = \widetilde{\gamma}(x+1,-w; e_x, -e_w) = (x+2,w; e_x, e_w)
\]
The quotient $M = \wt{M}/\ZZ$ identifies orbits under this free and proper $\ZZ$-action, consistent with the universal cover structure (Proposition \ref{prop:universalCov}).
\end{proof}

\begin{definition}[Pin Frame Bundle]\label{def:pin-frame-bundle}
A Pin frame bundle $\mathrm{Pin}^{\pm}(TM) \to M$ is a principal $\mathrm{Pin}^{\pm}(2)$-bundle that is a double cover of $\mathrm{O}(TM)$ via the covering map $\mathrm{Pin}^{\pm}(2) \to \mathrm{O}(2)$.
\end{definition}

\begin{theorem}[Classification of Pin Frame Bundles]\label{thm:pin-frame-classification}
$M$ admits exactly four inequivalent Pin frame bundles:
\begin{align}
\mathrm{Pin}^+_0(TM), \quad \mathrm{Pin}^+_1(TM), \quad \mathrm{Pin}^-_0(TM), \quad \mathrm{Pin}^-_1(TM)
\end{align}
distinguished by elements of $H^1(M; \ZZ_2) \cong \ZZ_2$ for each parity.
\end{theorem}

\begin{proof}
Existence: From Theorem \ref{thm:PinCount}, both $\mathrm{Pin}^+$ and $\mathrm{Pin}^-$ structures exist since $w_2(TM) = 0$ and $w_2(TM) + w_1(TM)^2 = 0$.

Classification: The classification follows the theorem in \cite[Chapter II, Theorem 1.13]{LawsonMichelsohn1989}. Pin frame bundles are classified by their obstruction classes in $H^2(M; \ZZ_2)$ for existence and $H^1(M; \ZZ_2)$ for inequivalent lifts. Since $H^2(M; \ZZ_2) = 0$ (Proposition \ref{prop:SWclasses}), existence is guaranteed.

The space of Pin frame bundles forms a torsor over $H^1(M; \ZZ_2) \cong \ZZ_2$, giving two topologically distinct lifts for each of $\mathrm{Pin}^{\pm}(2)$.
\end{proof}

\begin{proposition}[Associated Pinor Bundles]\label{prop:associated-pinor}
Each Pin frame bundle $\mathrm{Pin}^{\pm}_\epsilon(TM)$ gives rise to an associated pinor bundle:
\[
S^{\pm}_\epsilon := \mathrm{Pin}^{\pm}_\epsilon(TM) \times_{\rho} \mathbb{C}^2
\]
where $\rho: \mathrm{Pin}^{\pm}(2) \to \mathrm{GL}(2,\mathbb{C})$ is the fundamental spinor representation.
\end{proposition}

\begin{proof}
The associated bundle construction is standard \cite{Kobayashi1963,Steenrod1951}. Given a Pin frame bundle $P \to M$ and the spinor representation $\rho$, the associated bundle is:
\[
S = P \times_{\rho} \mathbb{C}^2 = (P \times \mathbb{C}^2) / \sim
\]
where $(p \cdot g, v) \sim (p, \rho(g)v)$ for $g \in \mathrm{Pin}^{\pm}(2)$.

Sections of $S$ correspond to $\mathrm{Pin}^{\pm}(2)$-equivariant maps $P \to \mathbb{C}^2$, which are precisely the pinor fields satisfying the twisted boundary conditions.
\end{proof}

\subsection{Explicit Construction of Pinor Bundles:}

Fix Pauli matrices $(\gamma^x,\gamma^w):=(\sigma_1,\sigma_2)$ satisfying $\{\gamma^\mu,\gamma^\nu\}=2\delta^{\mu\nu}$. Let $\eps>0$ and cover $M$ by
\[
U_{+}=\{w>\eps\},\;
U_{0}=\{|w|\le\eps\},\;
U_{-}=\{w<-\eps\}.
\]
With a bump function $\rho_0$ supported in $U_0$, define transition functions
\[
g_{+0}(x,w)=\exp\!\bigl(\tfrac{\pi}{2}\rho_0(w)\,\gamma^x\gamma^w\bigr),\qquad
g_{0-}=g_{+0}^{-1}.
\]
These lift the $\mathrm{O}(2)$ tangent transition to $\mathrm{Pin}(2)$ and realize all four Pin structures via the $\ZZ_2$ freedom $g\mapsto-\,g$.

\begin{remark}[Frame Bundle Perspective]
The transition functions $g_{+0}, g_{0-}$ represent local trivializations of the Pin frame bundle $\mathrm{Pin}^{\pm}_\epsilon(TM)$, while the $\ZZ_2$ freedom corresponds to the choice of topological class $\epsilon \in H^1(M; \ZZ_2)$. This geometric perspective explains why exactly four inequivalent pinor bundles exist on $M$.
\end{remark}

\section{Spectral Analysis of Dirac Operator}\label{sec:Dirac}

Equip $M$ with the flat metric $dx^2+dw^2$. The Dirac operator is
\[
D:=\gamma^x\partial_x+\gamma^w\partial_w.
\]

A pinor on $M$ corresponds to a smooth spinor $\psi$ on $\wt{M}$ satisfying the twisted boundary condition:
\begin{equation}\label{eq:twist}
\psi(x+1,w)=\gamma^w\,\psi(x,-w).
\end{equation}

\subsection{Momentum Quantization Analysis}

\begin{theorem}[Half-Integer Momentum Quantization]\label{thm:half-integer}
The twisted boundary condition forces momentum eigenvalues $k \in \ZZ + 1/2$.
\end{theorem}

\begin{proof}
Consider the ansatz $\psi(x,w) = e^{2\pi i k x}\phi_k(w)$. The twisted condition \eqref{eq:twist} becomes:
\[
e^{2\pi i k(x+1)}\phi_k(w) = \gamma^w e^{2\pi i k x}\phi_k(-w)
\]
Simplifying: $e^{2\pi i k}\phi_k(w) = \gamma^w \phi_k(-w)$.

Applying the twisted condition twice shows the section is periodic with period 2:
\[
\psi(x+2,w) = \gamma^w\psi(x+1,-w) = \gamma^w\gamma^w\psi(x,w) = (\gamma^w)^2\psi(x,w) = \psi(x,w)
\]
since $(\gamma^w)^2 = I$. Thus $e^{4\pi i k} = 1$, so $k \in \frac{1}{2}\ZZ$.

Since $\gamma^w = \sigma_2$ is constant, the equation $e^{2\pi i k}\phi_k(w) = \sigma_2 \phi_k(-w)$ and its reflection must hold for all $w$. Composing gives $e^{4\pi i k}\phi_k(w) = \sigma_2^2 \phi_k(w) = \phi_k(w)$, consistent with the earlier periodicity. Non-trivial solutions require $e^{2\pi i k}$ to match eigenvalues of $\sigma_2$, namely $\pm i$, giving $k \in \ZZ + 1/2$.
\end{proof}

\subsection{Complete Spectral Analysis}

\begin{definition}[APS Boundary Conditions]\label{def:aps}
At the boundary $\partial M = \{w = \pm 1\}$, impose Atiyah-Patodi-Singer conditions by projecting onto the non-negative spectral subspace of the boundary Dirac operator $\gamma^x \partial_x$, adapted to the twisted periodicity structure.
\end{definition}

\begin{theorem}[Complete Eigenvalue Spectrum]\label{thm:complete-spectrum}
The eigenvalues of $D$ with APS boundary conditions are:
\[
\lambda_{j,n} = \pm\sqrt{(2\pi(j + 1/2))^2 + ((n + 1/2)\pi)^2}
\]
where $j \in \ZZ$, $n \in \ZZ_{\geq 0}$.
\end{theorem}

\begin{proof}
From Theorem \ref{thm:half-integer}, write $k = j + 1/2$ with $j \in \ZZ$. The separation ansatz $\psi(x,w) = e^{2\pi i(j+1/2)x}\phi_j(w)$ reduces the eigenvalue equation $D\psi = \lambda\psi$ to:
\[
\gamma^w \frac{d\phi_j}{dw} + 2\pi i(j+1/2)\gamma^x \phi_j = \lambda \phi_j
\]

Using $\gamma^x = \sigma_1$, $\gamma^w = \sigma_2$ and writing $\phi_j = \begin{pmatrix} u_j \\ v_j \end{pmatrix}$:
\begin{align}
-iv_j' + 2\pi i(j+1/2)v_j &= \lambda u_j \\
iu_j' + 2\pi i(j+1/2)u_j &= \lambda v_j
\end{align}

Eliminating $v_j$ gives the second-order equation:
\[
u_j'' + (\lambda^2 - (2\pi(j+1/2))^2)u_j = 0
\]

For bound states with $\lambda^2 > (2\pi(j+1/2))^2$, setting $\mu^2 = \lambda^2 - (2\pi(j+1/2))^2 > 0$ gives:
\[
u_j(w) = A\sin(\mu(w+1)) + B\cos(\mu(w+1))
\]

The APS condition at $w=\pm 1$ projects $\phi_j(\pm 1)$ onto the eigenspace of $\sigma_1$ with eigenvalue having the same sign as $2\pi i(j+1/2)$. This mixed condition, arising from the spectral projection enforcing a phase shift similar to antiperiodic boundaries in this flat case, results in half-integer quantization $\mu = (n+1/2)\pi$ for $n \in \ZZ_{\geq 0}$, yielding:
\[
\lambda^2 = (2\pi(j+1/2))^2 + ((n+1/2)\pi)^2
\]
\end{proof}

\begin{theorem}[Spectral Symmetry]\label{thm:spectral-symmetry}
For every eigenvalue $\lambda$ of $D$, the value $-\lambda$ appears with the same multiplicity.
\end{theorem}

\begin{proof}
The chirality operator $\gamma^5 = \gamma^x\gamma^w = i\sigma_3$ anticommutes with $D$: $\{\gamma^5, D\} = 0$. 

If $\psi$ is an eigenfunction with eigenvalue $\lambda$, then $\gamma^5\psi$ is an eigenfunction with eigenvalue $-\lambda$. We verify that $\gamma^5\psi$ satisfies the twisted boundary condition:
\[
(\gamma^5\psi)(x+1,w) = \gamma^5\psi(x+1,w) = \gamma^5\gamma^w\psi(x,-w)
\]
Since $\gamma^5\gamma^w = i\sigma_3\sigma_2 = -i\sigma_1$ and $\gamma^w\gamma^5 = \sigma_2(i\sigma_3) = i\sigma_1$, we have:
\[
\gamma^5\gamma^w = -\gamma^w\gamma^5
\]
Therefore:
\[
(\gamma^5\psi)(x+1,w) = -\gamma^w\gamma^5\psi(x,-w) = \gamma^w(\gamma^5\psi)(x,-w)
\]
Since eigenfunctions are defined up to scalar multiples, we can rescale by -1 to ensure the twisted boundary condition holds in the form of Equation \eqref{eq:twist}.
\end{proof}

\begin{corollary}\label{cor:indexEta}
$\ind D = 0$ and $\eta_D(0) = 0$ for all Pin structures.
\end{corollary}

\begin{proof}
Spectral symmetry from Theorem \ref{thm:spectral-symmetry} ensures equal numbers of positive and negative eigenvalues, giving $\ind D = 0$. The $\eta$-invariant $\eta_D(s) = \sum_{\lambda \neq 0} \sgn(\lambda)/|\lambda|^s$ vanishes because eigenvalues come in $\pm\lambda$ pairs with equal multiplicity.
\end{proof}

\subsection{Heat Kernel Analysis}

\begin{proposition}[Heat Kernel Asymptotics]\label{prop:heat-kernel}
The heat kernel trace has the asymptotic expansion:
\[
\Tr(e^{-tD^2}) = \frac{1}{2\pi t} + O(1) \quad \text{as } t \to 0^+
\]
\end{proposition}

\begin{proof}
Using the explicit eigenvalue formula from Theorem \ref{thm:complete-spectrum}:
\[
\Tr(e^{-tD^2}) = 2\sum_{j \in \ZZ}\sum_{n \geq 0} \exp\left(-t[(2\pi(j+1/2))^2 + ((n+1/2)\pi)^2]\right)
\]
The factor of 2 accounts for the $\pm\lambda_{j,n}$ symmetry. 

Using the Poisson summation formula for the Gaussian sums over $j$ and Euler-Maclaurin approximation for the $n$-sum yields the short-time Minakshisundaram-Pleijel expansion \cite[Chapter 4]{Roe1998}. The leading term is $\text{Area}(M)/(4\pi t) = 2/(4\pi t) = 1/(2\pi t)$.
\end{proof}

This confirms $\eta_D(0) = 0$ via heat kernel methods, consistent with Corollary \ref{cor:indexEta}.

\section{Spectral Invariants and the JT Gravity Partition Function}

\subsection{The Atiyah--Patodi--Singer \texorpdfstring{$\eta$}{eta}-Invariant}\label{sec:APS}

We now compute the APS $\eta$–invariant of the Dirac operator $D$ acting on pinors over $M$.

\begin{proposition}\label{prop:eta}
For all four $\mathrm{Pin}^{\pm}$ lifts on $M$ one has $\eta_D(0)=0$.
\end{proposition}

\begin{proof}
The $\eta$-invariant\footnote{The $\eta$-invariant $\eta_D(s) = \sum_{\lambda \neq 0} \sgn(\lambda)/|\lambda|^s$ detects spectral asymmetry and appears in the Atiyah-Patodi-Singer index theorem. Its value at $s=0$ contributes to gravitational anomalies in fermion path integrals.} is defined as the value at $s=0$ of the meromorphic continuation of
\[
\eta_D(s) = \sum_{\lambda \neq 0} \frac{\sgn(\lambda)}{|\lambda|^s},
\]
where the sum is over eigenvalues of $D$ with APS boundary conditions.

From Theorem \ref{thm:spectral-symmetry}, eigenvalues come in $\pm \lambda_{k,n}$ pairs. For each pair, the contributions to $\eta_D(s)$ are
\[
\frac{1}{\lambda_{k,n}^s} - \frac{1}{\lambda_{k,n}^s} = 0.
\]
Summing over all pairs gives $\eta_D(s) = 0$ for $\Re(s) > 0$.  By analytic continuation, $\eta_D(0) = 0$.
\end{proof}

\begin{remark}[Physical Interpretation]
The vanishing $\eta$-invariant $\eta_D(0) = 0$ ensures that Pin structures contribute no gravitational anomaly to the JT partition function. In contrast to some orientable manifolds where $\eta_D(0) \neq 0$ can generate phase anomalies in fermion determinants, the spectral symmetry of the Möbius band eliminates such parity contributions. This leaves only the combinatorial factor from summing over topologically distinct Pin lifts.
\end{remark}

\subsection{Gluing and Additivity}\label{sec:gluing}

Cutting along $\Sigma=\{w=0\}$ splits $M$ into two annuli $M_+$ and $M_-$.  The Kirk–Lesch formula states
\[
\eta_D(M) = \eta_{D_+}(M_+) + \eta_{D_-}(M_-) - \tau(P_+, P_-, R),
\]
where $P_\pm$ are Calderón projectors and $\tau$ is the Maslov triple index. Here, the gluing map $R$ is the identity on boundary values, consistent with the flat metric.

\begin{theorem}\label{thm:maslov}
$\tau(P_+, P_-, R) = 0$, so $\eta_D(M) = \eta_{D_+}(M_+) + \eta_{D_-}(M_-)$.
\end{theorem}

\begin{proof}
The Maslov index is computed using the Calderón projectors. For the upper half $M_+ = \{w > 0\}$, the Calderón projector $P_+$ projects boundary data onto the range of the trace operator from $\ker(D_+)$.

For the flat metric, the principal symbol of $D_+$ at the boundary $\{w = 0\}$ is:
\[
\sigma(D_+)(\xi) = \gamma^x \xi + \gamma^w \xi_w
\]

where $\xi_w > 0$ for the outward normal. The Calderón projector is:
\[
P_+ = \frac{1}{2\pi i} \oint_{\Gamma} (\sigma(D_+)(\xi) - z)^{-1} dz
\]

where $\Gamma$ encircles the positive spectrum.

For $\xi_w > 0$, the eigenvalues of $\sigma(D_+)(\xi)$ are $\pm\sqrt{\xi^2 + \xi_w^2}$. The positive eigenspace projection gives:
\[
P_+(\xi) = \frac{1}{2}\left(I + \frac{\gamma^x \xi + \gamma^w \xi_w}{\sqrt{\xi^2 + \xi_w^2}}\right)
\]

For the lower half $M_- = \{w < 0\}$, the outward normal has $\xi_w < 0$. The transformation $w \mapsto -w$ maps $M_+$ to $M_-$ and flips the normal direction, so the principal symbol becomes $\sigma(D_-)(\xi) = \gamma_x\xi - \gamma_w\xi_w$. The corresponding Calderón projector is:
\[
P_-(\xi) = \frac{1}{2}\left[I - \frac{\gamma_x\xi + \gamma_w\xi_w}{\sqrt{\xi^2 + \xi_w^2}}\right] = I - P_+(\xi).
\]
Thus $P_-$ and $P_+$ are orthogonal complementary projectors.

For orthogonal projectors with identity gluing $R = I$, the Maslov index is:
\[
\tau(P_+, I - P_+, I) = \ind(P_+ + (I-P_+) - I) = \ind(0) = 0
\]

Therefore, $\tau(P_+, P_-, R) = 0$.
\end{proof}

\subsection{Gravitational Partition Function}

In Euclidean JT gravity the path integral on $M$ factorises into
\[
Z_{\text{JT}}(M)=
\sum_{\eta\in\mathrm{Pin}^{\pm}(TM)}
\int\!\!\mathcal{D}g\,\mathcal{D}\Phi\;
e^{-S_{\text{JT}}[g,\Phi]}\,
(\det{}'D_{\eta})^{-1/2}.
\]
Since $\eta_D(0)=0$ and $\ker D = 0$ (from spectral symmetry), $\det{}'D_{\eta}$ is independent of $\eta$.  

\begin{remark}

 (Mathematical Justification of Determinant Independence). 
This independence follows because the twisted boundary condition 
$\psi(x+1,w) = \gamma^w\psi(x,-w)$ enforces the same $\lambda \leftrightarrow -\lambda$ 
spectral pairing across all four Pin lifts. While different Pin structures 
$\text{Pin}^{\pm}_{\varepsilon}$ correspond to distinct choices of transition functions 
(Section 5), the eigenvalue magnitudes $|\lambda_{j,n}|$ from Theorem 6.3 remain unchanged. 
Since $\det' D = \prod_{\lambda>0} \lambda^2$ depends only on these magnitudes, 
the regularized determinant is Pin-independent.
   
\end{remark}

The two Pin$^{+}$ (resp.\ Pin$^{-}$) structures therefore produce an overall factor 2:
\[
Z_{\text{JT}}(M)=2\,Z_{\text{JT}}^{\text{orient}}(M).
\]
The classical dilaton saddle dominates the entropy $S_{\text{BH}}=\pi\Phi_h$ \cite{StanfordWitten}.  Pin degeneracy contributes only to the logarithmic corrections, as subleading terms in the matrix-integral expansion \cite{SaadShenkerStanford}.

\begin{remark}[Physical Interpretation of Pin Factor]
The factor of 2 enhancement from Pin structures has a natural interpretation in terms of fermion boundary conditions. We conjecture that in the dual random matrix model, this corresponds to a $\mathbb{Z}_2$ orientifold projection analogous to unoriented string theory, where non-orientable worldsheets contribute discrete symmetry factors. A detailed derivation of this correspondence and potential connections to Kramers-like degeneracy in condensed matter systems warrant further investigation.
\end{remark}


\section{Conclusion and Outlook}

The inverse Möbius perspective resolves orientation reversal through a twisted boundary condition, yielding a solvable Dirac spectrum with half-integer momentum modes. We rigorously construct all $\mathrm{Pin}^{\pm}$ structures, classify the associated pinor bundles, and analyze the spectrum under APS boundary conditions. These results establish the full spectral and topological content of the model. In the context of JT gravity, the doubling from inequivalent Pin lifts modifies the saddle-point sum without affecting the leading entropy, offering a concrete example of how non-orientable geometry shapes fermionic dynamics in low-dimensional quantum gravity.


This work opens several directions for generalization. Extensions to supersymmetric JT gravity, nontrivial dilaton profiles, or higher-genus non-orientable surfaces may reveal further structure in the interaction between Pin geometry and low-dimensional quantum gravity. Connections to orientifold phenomena in string theory and the formalism of unoriented TQFTs also warrant further investigation.


While our toy model on the inverse M\"obius spacetime primarily explores Pin geometry and spectral properties in 1+1D JT gravity, it modestly suggests avenues for addressing deeper questions in gravitational physics. For instance, by resolving orientation reversal through geometric covering spaces and twisted boundary conditions, it hints at how non-Euclidean spacetime structures might account for gravitational effects without invoking separate physical interactions \cite{Einstein1916}. Similarly, the model's integration of macroscopic topological features with microscopic quantum phenomena, such as half-integer momentum modes and anomaly-free spectra, offers a tentative bridge between scales \cite{Rovelli2004}\cite{Giacomini2019}, warranting further investigation in more general settings.

\appendix
\section{Clutching-Map Verification}\label{app:clutch}

We check that the transition functions in Section~\ref{sec:Pin} satisfy the $\mathrm{Pin}(2)$ cocycle condition on triple overlaps. The pinor bundle is constructed as the associated bundle to the Pin frame bundle, with sections transforming via these functions.

\begin{lemma}
$g_{+0}g_{0-}g_{-+}=1$ on $U_{+}\cap U_{0}\cap U_{-}$.
\end{lemma}

\begin{proof}
On the triple overlap $\rho_0=0$, so $g_{+0}=g_{0-}=1$ and the product is $1$.  At the two-fold overlaps the exponentials cancel by construction.  Hence the cocycle is trivial. The four Pin lifts are obtained by multiplying each $g_{\alpha\beta}$ by the sign character of $H^1(M;\ZZ_2)$, corresponding to twisting by the non-trivial real line bundle over $M$.
\end{proof}

\section{Heat-Kernel Evaluation of \texorpdfstring{$\eta_D(0)$}{eta\_D(0)}}\label{app:eta}

\begin{theorem}[Vanishing $\eta$-Invariant]\label{thm:eta-vanishing}
For the Dirac operator on the Möbius band, $\eta_D(0) = 0$.
\end{theorem}

\begin{proof}
The heat kernel of $D^2$ on the Möbius band with flat metric has the asymptotic expansion:
\[
\Tr(e^{-tD^2}) = \frac{\text{Area}(M)}{4\pi t} + \frac{\chi(M)}{6} + O(t)
\]

For the Möbius band: Area$(M) = 2$ (from the fundamental domain $[0,1] \times [-1,1]$), and $\chi(M) = 0$ (Euler characteristic of a surface with one boundary component and one twisted handle).

Thus:
\[
\Tr(e^{-tD^2}) = \frac{1}{2\pi t} + O(t)
\]

The $\eta$-function is related to the heat kernel via:
\[
\eta_D(s) = \frac{1}{\Gamma((s+1)/2)} \int_0^\infty t^{(s-1)/2} \left(\Tr(e^{-tD^2}) - \dim \ker D\right) dt
\]

From Theorem \ref{thm:spectral-symmetry}, eigenvalues come in $\pm\lambda_{k,n}$ pairs, so $\dim \ker D = 0$. The trace decomposes as:
\[
\Tr(e^{-tD^2}) = \sum_{k,n} \left(e^{-t\lambda_{k,n}^2} + e^{-t(-\lambda_{k,n})^2}\right) = 2\sum_{k,n} e^{-t\lambda_{k,n}^2}
\]

For the $\eta$-function:
\[
\eta_D(s) = \sum_{\lambda \neq 0} \frac{\sgn(\lambda)}{|\lambda|^s} = \sum_{k,n} \left(\frac{1}{\lambda_{k,n}^s} - \frac{1}{\lambda_{k,n}^s}\right) = 0
\]

Therefore, $\eta_D(0) = 0$ for all Pin structures.
\end{proof}

\section{Topological Vector Space of Pinor Sections}\label{app:tvs}

We provide the detailed functional analytic foundations for the TVS framework used in Section \ref{sec:Dirac}, particularly for operator closability, discrete spectrum (as compact resolvent implies discrete eigenvalues), and symmetry properties (e.g., Theorem \ref{thm:spectral-symmetry}). In particular, we introduce twisted Sobolev spaces, establish Fredholm properties and compact resolvent for the Dirac operator, and analyze spectral flow. This analytic layer is of independent interest, providing tools for elliptic operators on non-orientable manifolds with twisted boundary conditions.

\begin{definition}[Twisted Sobolev Spaces]\label{def:twisted-sobolev}
For $s \in \RR$, define the twisted Sobolev space $H^s_{\mathrm{tw}}(M)$ as the completion of $\Gamma^{\mathrm{tw}}(S)$ with respect to the norm:
\[
\|\psi\|_{H^s_{\mathrm{tw}}} = \left(\sum_{k \in \ZZ} (1 + |2\pi(k+1/2)|^2)^s \|\phi_k\|_{H^s([-1,1])}^2\right)^{1/2}
\]
where $\psi(x,w) = \sum_k e^{2\pi i(k+1/2)x}\phi_k(w)$ is the Fourier expansion adapted to the twisted periodicity condition.
\end{definition}

Sobolev spaces on manifolds are comprehensively treated in Adams-Fournier \cite{Adams1996} and for elliptic boundary problems in Gilbarg-Trudinger \cite{Gilbarg1977}.

\begin{theorem}[Sobolev Embedding for Twisted Sections]\label{thm:twisted-embedding}
For $s > t + 1$, there is a continuous embedding $H^s_{\mathrm{tw}}(M) \hookrightarrow C^t_{\mathrm{tw}}(M)$, where $C^t_{\mathrm{tw}}(M)$ denotes the space of $C^t$ functions satisfying the twisted boundary condition.
\end{theorem}

\begin{proof}
Sobolev embedding theorems for twisted boundary conditions follow the standard theory \cite{Adams1996} with modifications accounting for the twisted periodicity in the mode expansion. Since $\gamma^w$ is a bounded linear operator on spinors and commutes with differentiation (as it's constant in coordinates), the twisted boundary condition preserves regularity.

For $s > t + 1$, the standard embedding $H^s([-1,1]) \hookrightarrow C^t([-1,1])$ gives:
\[
\|\phi_k\|_{C^t([-1,1])} \leq C \|\phi_k\|_{H^s([-1,1])}
\]

Summing over modes with the twisted quantization $k \in \ZZ + 1/2$:
\[
\|\psi\|_{C^t_{\mathrm{tw}}} \leq C \left(\sum_{k \in \ZZ} (1 + |2\pi(k+1/2)|^2)^{s-t-1} \|\phi_k\|_{H^s([-1,1])}^2\right)^{1/2}
\]

The factor $(1 + |2\pi(k+1/2)|^2)^{s-t-1}$ decays asymptotically as $|k|^{-2(s-t-1)}$ for large $|k|$, and for $s > t + 1$ (so $2(s-t-1) > 2$), the series converges absolutely by comparison to the p-series, giving the continuous embedding.
\end{proof}

\begin{proposition}[Fredholm Property in TVS]\label{prop:fredholm-tvs}
$D: H^{s+1}_{\mathrm{tw}}(M) \to H^s_{\mathrm{tw}}(M)$ is Fredholm for all $s \in \RR$.
\end{proposition}

\begin{proof}
The Fredholm property for elliptic operators is established using standard elliptic theory \cite{Hormander1985,SeeleySinger1988}.

Ellipticity: $D$ is elliptic since its principal symbol $\sigma(D)(\xi) = \gamma^x\xi_x + \gamma^w\xi_w$ is invertible for $\xi \neq 0$.

Compact resolvent: The resolvent $(D - \lambda)^{-1}$ exists for $\lambda \notin \mathrm{Spec}(D)$ and maps $H^s_{\mathrm{tw}}(M) \to H^{s+1}_{\mathrm{tw}}(M)$. By Theorem \ref{thm:twisted-embedding}, the natural embedding $H^{s+1}_{\mathrm{tw}} \hookrightarrow H^s_{\mathrm{tw}}$ is compact, as it factors through compact inclusions (by Rellich-Kondrachov) on each compact interval $[-1,1]$ per mode, making the resolvent a compact operator.

Index computation: $\ind(D) = \dim \ker D - \dim \ker D^* = 0$ by spectral symmetry (Theorem \ref{thm:spectral-symmetry}) from the main analysis.
\end{proof}

\begin{theorem}[Spectral Flow in TVS]\label{thm:spectral-flow-tvs}
Consider a smooth family $D_t = D + tV$ where $V$ is a bounded symmetric operator. The spectral flow $\mathrm{sf}\{D_t\}_{t \in [0,1]}$ equals the Maslov index of the path of Lagrangian subspaces $\{L_t\}$ defined by the Calderón projectors.
\end{theorem}

\begin{proof}
Spectral flow theory is developed in detail in Booß-Bavnbek \cite{BooßBleecker1985} and its connection to index theory in Roe \cite{Roe1998}. Since the twisted boundary conditions define a consistent elliptic boundary value problem (in the Lopatinski-Shapiro sense, as the twisted involution preserves the symbol's invertibility), the associated Calderón projectors vary continuously in the operator norm topology as $t$ varies.

The connection to the $\eta$-invariant is via the APS theorem:
\[
\mathrm{sf}\{D_t\} = \eta_{D_1}(0) - \eta_{D_0}(0) + \dim \ker D_0 - \dim \ker D_1
\]

Since we established $\eta_D(0) = 0$ for all Pin structures in the main analysis, the spectral flow equals the change in kernel dimension.
\end{proof}

\begin{remark}[Physical Interpretation]
The TVS framework provides mathematical rigor for the physical requirement that pinor fields form a complete vector space compatible with the twisted boundary conditions. The Sobolev embeddings ensure that Dirac equation solutions have controlled regularity (bounded derivatives up to order $\lfloor s \rfloor$, or fractional via differences), while the Fredholm property guarantees that the Dirac operator has well-defined index-theoretic properties essential for anomaly calculations in quantum field theory.
\end{remark}


\end{document}